%% file: main.tex
\documentclass[submission]{eptcs}

\usepackage[dvips]{graphicx}

\usepackage{amsmath}
\usepackage{amsthm}
\usepackage{amssymb}
\usepackage{color}
\usepackage[utf8]{inputenc}
\usepackage{stmaryrd}
\usepackage{mathpartir}
\usepackage{enumitem}
\usepackage{cite}
\usepackage{array}
\usepackage{multirow}

\newtheorem{theorem}{Theorem}

\input{./definitions.tex}

\providecommand{\event}{EXPRESS/SOS 2013} 
\title{Towards Meta-Reasoning in the Concurrent Logical
  Framework CLF\thanks{This work was supported by the Qatar National Research Fund under grant
\mbox{NPRP 09-1107-1-168}.}}
\author{Iliano Cervesato \qquad\qquad Jorge Luis Sacchini
\institute{Carnegie Mellon University}
\email{~~iliano@cmu.edu\qquad\qquad~~ sacchini@qatar.cmu.edu}
}
\def\titlerunning{Towards Meta-Reasoning in CLF}
\def\authorrunning{I. Cervesato \& J.L. Sacchini}

\begin{document}
\maketitle
\begin{abstract}
  The concurrent logical framework CLF is an extension of the logical
  framework LF designed to specify concurrent and distributed languages.
  While it can be used to define a variety of formalisms, reasoning about such
  languages within CLF has proved elusive.  In this paper, we propose an
  extension of LF that allows us to express properties of CLF specifications.
  We illustrate the approach with a proof of safety for a small language with
  a parallel semantics.
\end{abstract}

\section{Introduction}\label{sec:introduction}

Due to the widespread availability of multi-core architectures and the growing
demands of web applications and cloud-based computation models, primitives for
programming concurrent and distributed systems are becoming essential features
in modern programming languages.  However, their semantics and meta-theory are
not as well understood as those of sequential programming languages.  This
limits our assurance in the correctness of the systems written in them.  Thus,
just as in the case of sequential languages 40 years ago, there has been
increasing interest in defining formal semantics that isolate and explain
their quintessential features.  Just as for sequential languages, such
semantics hold the promises of developing, for example, provably-correct
compilers and optimizations for such languages, as well as verification
frameworks for concurrent applications written using them.

Logical frameworks are formalisms designed to specify and reason about the
meta-theory of programming languages and logics. They are at the basis of tools
such as Agda~\cite{Norell07}, Coq~\cite{Coq10}, Isabelle~\cite{NipkowPW02}, and
Twelf~\cite{Twelf99}. The current generation of logical frameworks were designed
to study sequential programming languages, and specifying concurrent  systems
using these tools requires a large effort,
as the user is forced to define ad-hoc concurrency
models that are difficult to reuse and automate.

One way to deal with this problem is to design a logical framework that
natively embeds a general-purpose concurrency model.  This then provides
native support for describing parallel execution and synchronization, for
example, thus freeing the user from the delicate task of correctly encoding
them and proving properties about them.  One example of this approach is the
concurrent logical framework CLF~\cite{CervesatoPWW03, CervesatoWPW03,
  Schack-Nielsen11}, an extension of the logical framework LF~\cite{LF93}
designed for specifying concurrent, parallel, and distributed languages.  One
of its distinguishing features is its support for expressing concurrent
traces, i.e., computations where independent steps can be permuted.  For
example, traces can represent sequences of evaluation steps in a parallel
operational semantics, where executions that differ only in the order of
independent steps are represented by the same object (modulo permutation).
CLF has been used to encode a variety of systems such as Concurrent ML,
the $\pi$-calculus, and Petri nets in a natural way~\cite{CervesatoPWW03}.

However, unlike LF which permits specifying a system and its meta-theory
within the same framework, CLF is not expressive enough for proving
meta-theoretical properties about CLF specifications (e.g., type preservation,
or the correctness of program transformations).  The main reason is that
traces are not first-class values in CLF, and therefore cannot be manipulated.
In this work we propose a logical framework that supports meta-reasoning over
parallel, concurrent, and distributed specifications.  Specifically, the main
contributions of this paper are the following:
\begin{itemize}
\item%
  We define an extension of LF, called Meta-CLF, that allows meta-reasoning
  over a CLF specification.  It enriches LF with a type for concurrent traces
  and the corresponding constructor and destructors (via pattern-matching).
  This permits a direct manipulation of traces.  Meta-theorems can be
  naturally represented as relations, similar to the way sequential
  programming languages are analyzed in LF.
\item We illustrate the use of Meta-CLF by proving safety for a CLF
  specification of a small programming language with a parallel semantics.
\end{itemize}
The rest of the paper is organized as follows: in Sect.~\ref{sec:clf} we
recall CLF and use it to define the operational semantics of a simple parallel
language.  In Sect.~\ref{sec:meta-clf} we present Meta-CLF and use it to
express a proof of safety for this language.  We discuss related work in
Sect.~\ref{sec:related-work} and outline directions of future research in
Sect.~\ref{sec:conclusions}.

\section{CLF}
\label{sec:clf}

We begin by defining some key elements of CLF\@.  For conciseness, we omit
aspects of CLF that are not used in our examples.  The results of this paper
extend to the full language, however.  The presentation given here follows the
template proposed in~\cite{CervesatoPSSS12b} rather than the original
definition of CLF~\cite{CervesatoPWW03,CervesatoWPW03}; see
also~\cite{Schack-Nielsen11}.

\subsection{Syntax and Typing Rules}
\label{sec:syntax-typing-rules}

CLF is an extension of LF, or more precisely of the linear logical framework
LLF~\cite{cervesato02ic}, with a lax modality from lax
logic~\cite{FairtloughM97} used to encapsulate the effects of concurrent
computations.  The introduction form of lax modality are witnessed by a form
of proof term called {\em traces}.  A trace is a sequence of computational
steps where independent steps can be permuted.

The syntax of CLF is given by the following grammar:
\begin{align*}
  K &\deff \type \mid \GP\dec{{!}x}{T}.K & \text{(Kinds)} \\
  P &\deff a\Sapp{!}S \mid \{ \GD\} & \text{(Base types)} \\
  T &\deff \GP\dec{\cmod{x}}{T}.T \mid P & \text{(Types)} \\\ 
  N &\deff \Gl\dec{\cmod{x}}{T}.N \mid H\Sapp S \mid \{\eps\}  &
  \text{(Terms)} \\
  \eps &\deff \trempty \mid \eps_1;\eps_2 \mid \epat{\GD}{c\Sapp S} & \text{(Traces)} \\
  S &\deff \Semp \mid \cmod{N}; S & \text{(Spines)} \\
  \GD &\deff \Semp \mid \GD,\cmod{x}:T & \text{(Contexts)} \\
  \cmodname &\deff {\lin} \mid {!} & \text{(Modalities)}
\end{align*}
In this paper, we only consider CLF's persistent (!) and linear (\lin)
substructural modalities.  CLF also includes an affine modality, omitted here
for space reasons.

Kinds are as in LF\@.  Note that the argument in product kinds must be
persistent.  Base types are either {\em atomic} ($a\Sapp !S$) which are formed
by a constant applied to a persistent spine~\cite{CervesatoP03}, or {\em
  monadic} which are a context enclosed in the lax modality, denoted with
$\{\_\}$.

A type is either a product or a base type.  We consider two different
products: a persistent product, $\GP\dec{{!}x}{T}{U}$, as in LF, and a linear
product, $\GP\dec{{\lin}x}{T}{U}$, from LLF\@.  Note, however, that the typing
rules prevent dependencies on linear products.  We usually write $T\to U$ and
$T\lolli U$ for $\GP\dec{{!}x}{T}{U}$ and $\GP\dec{{\lin}x}{T}{U}$,
respectively, when $x$ is not free in $U$.

A term is either an abstraction (persistent and linear), an atomic term $H\Sapp
S$ formed by a variable $H$ applied to a list of arguments given by the spine
$S$, or a trace $\{\eps\}$.

A trace is either the empty trace ($\trempty$), a composition of traces
($\eps_1;\eps_2$) or an individual step of the form $\epat{\GD}{c\Sapp S}$,
where $c$ is a constant defined in the signature applied to a spine $S$, whose
type must be monadic.  This step {\em consumes} the linear variables in $S$
and {\em produces} the linear and persistent variables in $\GD$.  A step binds
the variables defined in $\Delta$ in any trace that follows it.

Concurrent computation is expressed by endowing traces with a monoidal
structure: the empty trace is the unit, and trace composition is associative.
Furthermore, it allows permutation of independent steps.  Step independence is
defined on the basis of the notion of {\em trace interface}.  The {\em input
  interface} of a trace, denoted $\ein{\eps}$, is the set of variables used by
$\eps$, i.e., its free variables.  The {\em output interface} of a trace,
denoted $\eout{\eps}$, is the set of variables defined by $\eps$.  They are
given by the following equalities:
\begin{align*}
     \ein{(\trempty)} &\;=\; \emptyset
     & \eout{(\trempty)} &\;=\; \emptyset
  \\ \ein{(\epat{\GD}{c\Sapp S})} &\;=\; \FV{S}
  & \eout{(\epat{\GD}{c\Sapp S})} &\;=\; \dom{\GD}
  \\ \ein{(\eps_1;\eps_2)} &\;=\; \ein{\eps_1}\cup(\ein{\eps_2}\setminus\eout{\eps_1})
   & \eout{(\eps_1;\eps_2)} &\;=\;
  \eout{\eps_2}\cup(\eout{\eps_1}\setminus\ein{\eps_2})\cup{!}{(\eout{\eps_1})}
\end{align*}
where $\FV{S}$ is the set of free variables in $S$, and $\dom{\GD}$ is the set
of variables declared in $\GD$.  In a trace composition, the output interface
contains all the persistent variables introduced in $\eps_1$, even if $\eps_2$
uses them.  In other words, persistent facts cannot be removed from the output
once they are introduced.  On the other hand, linear facts are effectively
removed if $\eps_2$ uses them.

Two traces $\eps_1$ and $\eps_2$ are {\em independent}, denoted
$\eps_1\parallel \eps_2$, if $\ein{\eps_1}\cap\eout{\eps_2}=\emptyset$ and
$\eout{\eps_1}\cap\ein{\eps_2}=\emptyset$.  Independent traces do not share
variables and can therefore be executed in any order.

\paragraph{Equality.}
We denote with $\eeq$ the equality relation on kinds, types, terms, spines,
and contexts.  It is defined as $\alpha$-equality extended with trace
equality, also denoted with $\eeq$, defined by the following
rules:
$$
\infer{ }{\eps;\trempty\eeq\eps} \qquad
\infer{ }{\eps\eeq\eps;\trempty}
\qquad
\infer{ }{\eps_1;(\eps_2;\eps_3)\eeq(\eps_1;\eps_2);\eps_3}
$$
$$
\infer{\eps_1\parallel\eps_2}{\eps_1;\eps_2\eeq\eps_2;\eps_1}
\qquad
\infer{\eps_1\eeq\eps_1'}{\eps_1;\eps_2\eeq\eps_1';\eps_2}
\qquad
\infer{\eps_2\eeq\eps_2'}{\eps_1;\eps_2\eeq\eps_1;\eps_2'}
$$
These rules state that traces form a monoid and that independent steps can be
permuted.

\paragraph{Typing.}
The typing rules of CLF rely on some auxiliary meta-level operators.  We say
that a context $\GD$ {\em splits} into $\GD_1$ and $\GD_2$, denoted
$\GD=\GD_1\cjoin\GD_2$, if each persistent declaration in $\GD$ appears in
both $\GD_1$ and $\GD_2$, and each linear declaration in $\GD$ appears in
exactly one of $\GD_1$ and $\GD_2$.


We write $X[N/x]$ for the {\em hereditary substitution} of variable $x$ by
term $N$ in $X$ (where $X$ belongs to one of the CLF syntactic classes).
Hereditary substitutions~\cite{watkinscpw03} normalizes a terms as the
substitution is carried out, thereby allowing us to restrict the definition of
CLF to canonical terms (in $\beta\eta$-normal form).  Its definition is
type-directed and therefore terminating.  The interested reader can find an
exhaustive account in~\cite{watkinscpw03}.

The typing rules of CLF are displayed in Fig.~\ref{fig:typing-clf}.  We assume
a fixed signature $\GS$ of global declarations for kinds and types.  Typical
of type theories, we use bidirectional typing rules where types of variables
are inferred from the context, while terms are checked against a type.

The rules for kinds, types, spines and terms are standard~\cite{watkinscpw03}.
Note that only persistent variables are dependent in types.  The typing rules
for traces show the intuition that a trace is a context transformer: we can
read the judgment $\jud{\GD}{\eps}{\GD'}$ as ``\emph{$\eps$ transform the context
$\GD$ into $\GD'\,$}''.  Note that the trace typing rules imply a form of the
frame rule: in fact, it is easy to prove that if $\jud{\GD_1}{\eps}{\GD_2}$,
then $\jud{\GD_0\cjoin\GD_1}{\eps}{\GD_0\cjoin\GD_2}$.  The empty trace does
not change the context, while traces can be composed if the internal interface
matches.  For a single step, part of the context is transformed: the spine $S$
consumes $\GD_1$ and generates the context $\GD'$ (or equivalently, $\GD_2$).


\begin{figure}
\noindent Contexts: \fbox{$\sjud{{!}\GD}{\GD'}$}
$$
\infer
  { }
  {\sjud{{!}{\GD}}{\Semp}}
\qquad
\infer
  {\jud{{!}{\GD}}{T}{\type} \\
    \sjud{{!}{(\GD,\dec{\cmod{x}}{T})}}{\GD'} }
  {\sjud{{!}{\GD}}{\dec{\cmod{x}}{T},\GD'}}
$$
Kinds: \fbox{$\jud{{!}\GD}{K}{\kind}$}
$$
\infer
  { }
  { \jud{{!}\GD}{\type}{\kind} }
\qquad
\infer
  { \jud{{!}\GD}{T}{\type} \\ \jud{{!}\GD,\dec{{!}x}{T}}{K}{\kind} }
  { \jud{{!}\GD}{\GP\dec{{!}x}{T}.K}{\kind} }
$$
Base types: \fbox{$\jud{{!}\GD}{P}{K}$}
$$
\infer
  { a:\GP{!}\GD'.\type \\ \jud{{!}\GD}{S}{{!}\GD'} }
  { \jud{{!}\GD}{a\Sapp S}{\type} }
\qquad
\infer
  { \sjud{{!}\GD}{\GD'} }
  { \sjud{{!}\GD}{\{\GD'\}} }
$$
Types: \fbox{$\jud{{!}\GD}{T}{K}$}
$$
\infer
  { \jud{{!}\GD}{T}{\type} \\ \jud{{!}\GD,\dec{{!}x}{T}}{T}{\type} }
  { \jud{{!}\GD}{\GP\dec{{!}x}{T}.T}{\type} }
\qquad
\infer
  { \jud{{!}\GD}{T}{\type} \\ \jud{{!}\GD}{U}{\type} \\ \text{$x$ fresh}}
  { \jud{{!}\GD}{T\lolli U}{\type} }
$$
Terms: \fbox{$\jchk{\GD}{N}{T}$}
$$
\infer{ \jchk{\GD,\dec{\cmod{x}}{T}}{N}{U} }
      { \jchk{\GD}{\Gl\dec{\cmod{x}}.N}{\GP\dec{\cmod{x}}{T}.U} }
\qquad
\infer{ \jud{\GD}{\eps}{\GD'} }
      { \jchk{\GD}{\{\eps\}}{\{\GD'\}} }
$$
$$
\infer{ \dec{{!}{x}}{T} \in  \GD \\
        \jsp{\GD}{S}{T}{U} \\ U \eeq U'}
      { \jchk{\GD}{x\Sapp S}{U'} }
\qquad
\infer{ \jsp{\GD_0,\GD_1}{S}{T}{U} \\ U \eeq U'}
      { \jchk{\GD_0,\dec{\lin x}{T},\GD_1}{x\Sapp S}{U'} }
$$
$$
\infer{\dec{c}{T} \in  \GS \\
       \jsp{\GD}{S}{T}{U} \\ U \eeq U' }
      {\jchk{\GD}{c\Sapp S}{U'}}
$$
Traces: \fbox{$\jud{\GD}{\eps}{\GD'}$}
$$
\infer{ }{\jud{\GD}{\trempty}{\GD}}
\qquad
\infer
  {\jud{\GD}{\eps_1}{\GD_1} \\
    \jud{\GD_1}{\eps_2}{\GD_2}}
  {\jud{\GD}{\eps_1;\eps_2}{\GD_2}}
$$
$$
\infer
  {\dec{c}{T} \in  \GS
    \\ \jsp{\GD_1}{S}{T}{\{\GD'\}}
    \\ \GD_2 \eeq \GD' }
  { \jud{\GD_0\cjoin\GD_1}
        {\epat{\GD_2}{c\Sapp S}}
        {\GD_0,\GD_2} }
$$
Spines: \fbox{$\jsp{\GD}{S}{T}{T'}$}
$$
\infer
  { }
  { \jsp{\GD}{\Semp}{T}{T} }
\qquad
\infer
  { \jchk{{!}\GD_1}{N}{T} \\ \jsp{\GD_0}{S}{T_2[N/x]}{T'} }
  { \jsp{\GD_0\cjoin{!}\GD_1}{({!}N;S)}{\GP\dec{{!}x}{T}{U}}{T'} }
\qquad
\infer
  { \jchk{{\lin}\GD_1}{N}{T} \\ \jsp{\GD_0}{S}{U}{T'} }
  { \jsp{\GD_0\cjoin{\lin}\GD_1}{({\lin}N;S)}{T\lolli U}{T'} }
$$
\caption{CLF typing rules}
\label{fig:typing-clf}
\end{figure}

\subsection{Substructural Operational Semantics in CLF}
\label{sec:substr-oper-semant}

In this section, we show how to use CLF to define a substructural operational
semantics (SSOS) for a programming language.  As a case study, we illustrate
the approach on the simply-typed lambda calculus with an operational semantics
that evaluates functions and their arguments in parallel.

SSOS specifications have two main features.  The first is compositionality:
extending a programming language with a new feature does not invalidate the
SSOS already developed~\cite{Simmons12}.  Second, SSOS specifications can
naturally express parallel and concurrent semantics of a programming
language~\cite{Pfenning04}.

The language of expressions and types for the simply-typed lambda calculus,
\stlc, is given by
the following grammar:
\begin{align*}
  e &\deff x \mid \Gl x.e \mid e\,e & \text{(Expressions)} \\
  \Gt &\deff o \mid \Gt \to \Gt & \text{(Types)}
\end{align*}
Expressions are either variables, abstractions or applications; types are
either base types or function types.  Typing is defined by the judgment
$\GG\vdash e:\Gt$, given by the following rules:
$$
\infer
  { x:\Gt \in  \GG }
  { \GG\vdash x:\Gt }
\qquad
\infer
  { \GG,x:\Gt_1\vdash e:\Gt_2 }
  { \GG\vdash \Gl x.e: \Gt_1\to\Gt_2 }
\qquad
\infer
  { \GG\vdash e_1:\Gt_2\to\Gt_1 \\ \GG\vdash e_2:\Gt_2 }
  { \GG\vdash e_1\,e_2: \Gt_1 }
$$
Evaluation is given by $\beta$-reduction: $(\Gl x.e_1)\,e_2 \rewto e_1[e_2/x]$.

In CLF (and LF as well), we can represent this language and its typing rules
using a higher-order abstract syntax encoding as follows.
For clarity, we use implicit arguments (which can be reconstructed):

\begin{align*}
  \texp &: \type &
  \tof &: \texp\to\tp\to\type \\
  \lam &: (\texp\to\texp) \to\texp &
  \qquad\tofapp &: \tof~(\app~e_1~e_2)~t_1
  \\
  \app &: \texp\to\texp\to\texp &
  & ~~\gets \tof~e_1~(\arr~t_2~t_1) \\
  \tp &: \type &
  & ~~\gets \tof~e_2~t_2 \\
  \arr &: \tp \to\tp\to\tp &
  \toflam &: \tof~(\lam~e_2)~(\arr~t_1~t_2) \\
  \val &: \texp\to\type &
  & ~~\gets (\Pi x:\texp.~\tof~x~t_1\to\tof~(e_2~x)~t_2) \\
  \vallam &: \val~(\lam~\Gl x.e~x)
\end{align*}

We denote this signature with $\GS_\stlc$.
The syntax of \stlc{} is defined by the type $\texp$ with constructors $\app$
(representing function application) and $\lam$ (representing abstraction).
The type $\tp$ encodes types from \stlc\ with $\arr$ being the function type
from \stlc.
Variables are implicitly defined using CLF variables.
Similarly, there is no explicit representation of the context for typing;
instead, we use CLF's (persistent) context for this purpose.
The typing relation is expressed by the CLF type family $\tof$ relating \stlc\
expressions and types.
We also define the predicate $\val$ stating that abstractions are
values.

In the following, we define a SSOS for \stlc\ evaluation using
destination-passing style~\cite{Pfenning04}.
The SSOS is given by a state and a set of rewriting rules.
The state is a multiset whose elements have one of the following forms:
\begin{itemize}
\item%
  $\eval~e~d$: evaluate expression $e$ in destination $d$.  Destinations are
  virtual locations that store expressions to be evaluated and results.
\item%
  $\ret~e~d$: the result $e$ of an evaluation is stored at $d$. An invariant
  of the semantics ensures that $e$ is always a value.
\item%
  $\fapp~d_1~d_2~d$: an application frame expecting the result of evaluating a
  function in $d_1$, its argument in $d_2$, and storing the result in $d$.
\end{itemize}
The rewriting rules encoding expression evaluation are as follows:
\begin{align*}
   \eval~e~d   &~~\rewto~~ \ret~e~d
 & \text{if $\mathsf{value}(e)$}
\\ \eval~(e_1\,e_2)~d  &~~\rewto~~\eval~e_1~d_1,\;\;\eval~e_2~d_2,\;\;\fapp~d_1~d_2~d
 & \text{$d_1,d_2$ fresh}
\\ \ret~(\Gl x.e_1)~d_1,\;\;\ret~e_2~d_2,\;\;\fapp~d_1~d_2~d &~~\rewto~~ \eval~(e_1[e_2/x])~d
\end{align*}
The first rule says that evaluating a value expression immediately returns the
result.  The second rule says that to evaluate an application, we evaluate the
function and argument in two fresh destinations, and create a frame that
connects the results.  Evaluation of function and argument can proceed in
parallel.  The third rule computes the application once we have the values of
the function and argument connected by a frame.

A complete evaluation of an expression $e$ to a value $v$ is given by a
sequence of multisets of the form:
$$
\MA_0=\{ \eval~e~d \} \rewto \MA_1 \rewto … \rewto \MA_{n-1} \rewto \MA_n=\{\ret~v~d\}
$$
where at each step $\MA_i\rewto\MA_{i+1}$ part of $\MA_i$ is rewritten using one of the
evaluation rules given above.

This semantics can be faithfully represented in CLF using the linear context to
represent the multiset state, where each element is a linear fact, and
destinations are represented by persistent facts.
The semantics is given by the following CLF signature:
\begin{align*}
  \dest &: \type \\
  \eval &: \texp\to\dest\to\type \\
  \ret &: \texp\to\dest\to\type \\
  \fapp &: \dest\to\dest\to\dest\to\type \\
  \stepeval &: \eval~e~d\lolli\val~e\to\{ \dec{\lin x}{\ret~e~d} \} \\
  \stepapp &: \eval~(\app~e_1~e_2)~d \lolli \{ \dec{{!}d_1}{\dest},
  \dec{{!}d_2}{\dest}, \dec{\lin x_1}{\eval~e_1~d_1},
  \dec{\lin x_2}{\eval~e_2~d_2}, \dec{\lin x_3}{\fapp~d_1~d_2~d} \} \\
  \stepbeta &: \ret~(\lam~e_1)~d_1 \lolli~\ret~e_2~d_2\lolli\fapp~d_1~d_2~d\lolli
  \{ \dec{\lin x}{\eval~(e_1~e_2)~d} \}
\end{align*}
We denote with $\GS_\step$ the signature containing the declarations of the
evaluation rules.
Rule $\stepeval$ is effectively a conditional rewriting rule.
In rule $\stepapp$, new destinations ($d_1$ and $d_2$) are created to evaluate
function and argument; these evaluations can proceed in parallel.

\paragraph{Safety.}
Safety for this language is proven by giving a suitable notion of what a valid
state looks like~\cite{Simmons12}.  Note that not all multisets are valid
state.  For example, the singleton $\{ \fapp~d_1~d_2~d \}$ is not valid, since
there is no expression to evaluate at $d_1$ or $d_2$; similarly $\{
\eval~e_1~d,\eval~e_2~d \}$ is not valid since there are two expressions
evaluating on the same destination.

In a valid state, the elements should form a tree whose nodes are linked by
destinations.  Internal nodes have the form $\fapp~d_1~d_2~d$, with two
children (corresponding to $d_1$ and $d_2$) and the leaves are of the form
$\eval~e~d$ or $\ret~v~d$ (where $v$ is a value).  Furthermore, the types of
the expressions should match.

Following Simmons~\cite{Simmons12}, we define well-typed states by rewriting
rules.  The idea is to create the tree {\em top-down} starting at the root.
We can write these rules in CLF as follows:
\begin{align*}
  \gen&:\tp\to\dest\to\type \\
  \geneval &: \gen~t~d\lolli~\tof~e~t\to\{\dec{\lin x}{\eval~e~d}\} \\
  \genret &: \gen~t~d\lolli~\tof~e~t\to\val~e~\to\{\dec{\lin x}{\ret~e~d}\} \\
  \genfapp &: \gen~t~d \lolli \{ \dec{{!}d_1}{\dest}, \dec{{!}d_1}{\dest},
  \dec{\lin x_0}{\fapp~d_1~d_2~d}, \dec{\lin x_1}{\gen~(\arr~t_1~t)~d_1},
  \dec{\lin x_2}{\gen~t_1~d_2} \} \\
  \gendest &: \{ \dec{{!}d}{\dest} \}
\end{align*}
We denote with $\GS_\gen$ the signature containing these generation rules.
A fact of the form $\gen~t~d$ is read as {\em ``generate a tree with root at
  destination $d$ and type $t$''}.
We have three ways to do this: by generating a leaf of either the form
$\eval~e~d$ or $\ret~e~d$ (for an $e$ of the appropriate type), or by generating
an internal node $\fapp~d_1~d_2~d$ and then generating trees rooted at $d_1$ and
$d_2$.
Rule $\gendest$ is necessary to keep track of destinations that were created
during evaluation (by rule $\stepapp$) but are not used anymore (after the
application is reduced using $\stepbeta$, the destinations created to evaluate
the function and the argument are not needed anymore; see Sect.~\ref{sec:safety-ssos-spec}).

A generic tree is built by a sequence of rewriting steps starting from a
single $\gen~t~d$: $\MA_0=\{ \gen~t~d \} \rewto \MA_1 \rewto … \rewto \MA_n$,
where $\MA_n$ does not contain facts of the form $\gen~t~d$.

This kind of {\em generative rules} for describing valid states generalizes
context-free grammars, with $\gen$ being a non-terminal, and $\eval$, $\ret$,
and $\fapp$ are terminal symbols~\cite{Simmons12}.
Generative rules (also called generative grammars) are very powerful allowing to
express a wide variety of invariants.

With this definition of a well-typed state, we can prove that the language is
safe.  Safety is given by two properties: type preservation
(i.e., evaluation preserves well-typed states), and progress (i.e., either the
state contains the final result, or it is possible to make a step), as
stated in the following theorem.
We write $\MA\rewto^*_\GS\MA'$ to mean a {\em maximal} rewrite sequence from
$\MA$ to $\MA'$ using the rules in $\GS$. The sequence is maximal in the sense
that $\MA'$ does not contain non-terminal symbols.
We write $\MA\rewto^1_\GS\MA'$ to mean a {\em step} from $\MA$ to $\MA'$ using
one of the rules in $\GS$.

\begin{theorem}\label{thm:safety-stlc}
  The language defined by the signatures $\GS_\stlc$, $\GS_\step$ and
  $\GS_\gen$ is safe, i.e., it satisfies the following properties:
\begin{description}
\item[Preservation] If $\{ \gen~t~d \} \rewto^*_{\GS_\gen} \GD$ and $\GD
  \rewto^1_{\GS_\step} \GD'$, then $\{ \gen~t~d \} \rewto^*_{\GS_\gen} \GD'$.
\item[Progress] If $\{ \gen~t~d \} \rewto^*_{\GS_\gen} \GD$, then either $\GD$ is
  of the form $\{\ret~v~d\}$ with $\val~v$, or there exists $\GD'$ such that
  $\GD\rewto_{\GS_\step}\GD'$.
\end{description}
\end{theorem}
\begin{proof}[Proof sketch]
  (See~\cite{Simmons12} for details.) Preservation proceeds by case analysis on
  the rewriting step and inversion on the generated trace. Let us consider the
  case $\stepeval$. We have that $\GD$ must be of the form $\GD_0,{\eval~e~d}$,
  and $\GD'$ must be of the form $\GD_0,{\ret~e~d}$, where $\val~e$. Then, the
  generation trace for $\GD$ must contain a step using $\geneval$ to construct
  $\eval~e~d$. The generation trace for $\GD'$ is constructed by replacing this
  step with a $\genret$ step.

  Progress proceed by induction on the length of the generating trace and case
  analysis on the first step.  The interesting case is when this step is
  $\genfapp$.  The rest of the trace can be split in two parts, one generating
  trace for each child of the $\fapp$ generated in the first step.  The proof
  follows by induction on these subtraces.
\end{proof}


\section{Meta-CLF}
\label{sec:meta-clf}

Meta-theorems such as preservation and progress cannot be expressed in CLF,
since it lacks primitives for manipulating traces as first-class objects.  For
example, we cannot talk about the type of generated traces of the form
$$ \{\GD \} \rewto^*_{\GS_\gen} \{ \GD' \} $$
which is essential to express preservation.

Furthermore, CLF lacks abstractions over context, which prevents us from
defining a trace type that is parametric over its interface.
For example, in CLF we can define a relation on traces
$$ {\sf rel}: (A \lolli \{ \dec{\lin x}{B} \}) \to
(A \lolli \{ \dec{\lin x}{B} \}) \to\type $$ that relates two traces that
transform an $A$ into a $B$.  However, we cannot define the type of traces as
a transformation between two generic contexts $\GD_1$ and $\GD_2$.  For this,
we need to quantify over contexts.  With a dependent product that takes
contexts as arguments, we can define of all traces that generate valid states
starting from a seed:
$$
\Pi t:\tp. \Pi \Gp:\ctx.~ (\Pi \dec{{!}d}{\dest}. \gen~d~t\lolli \{ \Gp \})
$$

We use these ideas for designing a logical framework that permits
meta-reasoning on CLF specifications.  The resulting framework, which we call
Meta-CLF, is an extension of LF with trace types and quantification over
contexts.  Trace types have the form
$$
\trtype{\GD}{\GS}{\GD'}
$$
where $\GD$ and $\GD'$ are CLF contexts and $\GS$ is a CLF signature that
contains (monadic) rewriting rules (e.g., $\GS_\gen$ and $\GS_\step$).  A term
of this type is a trace of the form:
$$
\Gd_1; \ldots; \Gd_n
$$
where each step $\Gd_i$ is either $\epat{\GD_i}{c_i\Sapp S_i}$ with $c_i$
declared in $\GS$, or $x_i\Sapp S_i$ where $x_i$ is a (Meta-CLF) variable
that, applied to $S_i$, returns a trace.  The interface of the whole trace is
given by $\GD$ and $\GD'$.

Meta-CLF includes two different trace types: $\trtype{\GD}{\GS^*}{\GD'}$ and
$\trtype{\GD}{\GS^1}{\GD'}$.  The former defines maximal traces,
while the latter defines traces of exactly one step.  We write
$\trtype{\GD}{\GS}{\GD'}$ to refer to either of these types.

\subsection{Syntax and Typing Rules}
Meta-CLF is an extension of LF with trace and context types, parameterized
over a CLF signature, denoted $\GS_0$.  The kinds, types, and contexts of
Meta-CLF are given by the following grammar:
\begin{align*}
  K &\deff \type \mid \GP\dec{x}{A}.K \mid \GP\dec{\Gp}{\ctx}.K
  \mid \clfpi\dec{x}{T}.{K} \mid \nab x. K & \text{(Kinds)} \\
  A &\deff a\Sapp S \mid \GP\dec{x}{A}.A \mid \GP\dec{\Gp}{\ctx}.A
  \mid \clfpi\dec{x}{T}.{A} \mid \nab x.A
  & \multirow{0}{*}{\text{(Types)}} \\
  & ~~~~\mid \trtype{\GD}{\GS^*}{\GD} \mid \trtype{\GD}{\GS^1}{\GD} \\
  \GD &\deff \Semp \mid \Gp, \GD \mid \dec{\lin x}{A}, \GD \mid \dec{{!}x}{A},
  \GD &\text{(Contexts)}
\end{align*}
Kinds include, besides the constructions derived from LF, products over
contexts ($\Pi\dec{\Gp}{\ctx}.K$), products over CLF types from the signature
$\GS_0$ ($\clfpi\decx{T}{K}$), and products over names ($\nab x.K$).  In
$\clfpi\decx{T}{K}$, the type $T$ must be well-typed in the signature $\GS_0$,
using the CLF typing rules.  Quantification over names is necessary for
composing traces, to ensure that names declared in the interface match.
The notation is taken from~\cite{MillerT03}.

Types in Meta-CLF include the analogous kinds contructors applied to the type
level, with the addition of the trace types $\trtype{\GD}{\GS^*}{\GD'}$ and
$\trtype{\GD}{\GS^1}{\GD'}$.
Contexts are sequences consisting of linear declarations, persistent
declarations, and context variables, whose types are CLF types checked in the
signature $\GS_0$.
Names of declared variables must be introduced using $\nab$.

Type preservation for the language described in
Sect.~\ref{sec:substr-oper-semant} can then be stated in Meta-CLF as follows:
\begin{multline*}
  \clfpi\dec{t}{\tp}.~\nab d.~\nab g.~\Pi\dec{\Gp_1}{\ctx}.~\Pi\dec{\Gp_2}{\ctx}.~
\\\trtype{{!}d:\dest,\lin g:\gen~d~t}{\GS_\gen^*}{\Gp_1} \to
  \trtype{\Gp_1}{\GS_\step^1}{\Gp_2} \to
  \trtype{{!}d:\dest,\lin g:\gen~d~t}{\GS_\gen^*}{\Gp_2} \to \type
\end{multline*}
This type family can be read functionally as follows: given a trace that
generates $\Gp_1$ (using $\GS_\gen$), and a step from $\Gp_1$ to $\Gp_2$, we can
obtain a trace that generates $\Gp_2$ (using $\GS_\gen$).

Terms, spines, and traces in Meta-CLF are defined by the following grammar:
\begin{align*}
  N &\deff \Gl x.N \mid H\Sapp S \mid \Gl\Gp.N \mid \clflam x.N \mid \{\eps\} & \text{(Terms)}\\
  H &\deff x \mid c & \text{(Heads)} \\
  S &\deff \Semp \mid N; S \mid \GD; S \mid \#; S \mid \clf{M} &\text{(Spines)}
  \\
  \eps &\deff \diamond \mid \eps_1;\eps_2 \mid \epat{\GD}{c\Sapp  S} \mid x\Sapp S & \text{(Traces)}
\end{align*}
Terms include the introduction forms of Meta-CLF type ($\Gl x.N$), contexts
($\Gl\Gp.N$), CLF types ($\clflam x.N$), as well as atomic terms (a variable
applied to a spine) and traces ($\{\eps\}$).
Spines are sequence formed by terms ($N$), contexts ($\GD$), CLF terms
($\clf{M}$), and fresh names ($\#$).

A trace is either empty ($\trempty$), a composition of two traces
($\eps_1;\eps_2$), a single step ($\epat{\GD}{c\Sapp S}$) where $c$ is defined
in the CLF signature $\GS_0$, or a (Meta-CLF) trace variable applied to a spine
($x\Sapp S$).

\paragraph{Typing rules.}
Typing judgments are parameterized by the CLF signature $\GS_0$.
The typing rules are defined by the following judgments:
\begin{align*}
  & \jkind[\GS_0]{\GS;\GG;\GX}{K} & \text{(Kinds)} \\
  & \jtype[\GS_0]{\GS;\GG;\GX}{A}{K} & \text{(Types)} \\
  & \jctx[\GS_0]{\GS;\GG;\GX}{\GD} & \text{(Contexts)} \\
  & \jterm[\GS_0]{\GS;\GG;\GX}{M}{A} & \text{(Terms)} \\
  & \jsp[\GS_0]{\GS;\GG;\GX}{S}{A'}{A} & \text{(Spines)} \\
  & \jterm[\GS_0]{\GS;\GG;\GX}{\eps}{\trtype{\GD}{\GS}{\GD'}} & \text{(Traces)}
\end{align*}
The signature $\GS$ is a Meta-CLF signature, while $\GS_0$ is a CLF signature.
We usually omit them for clarity.
The context $\GG$ contains the declarations of Meta-CLF variables, contexts, and
CLF variables.
The context $\GX$ contains name declarations.

These contexts are defined by the following grammar:
\begin{align*}
\GG &\deff \Semp \mid \GG,\decx{A} \mid \GG,\dec{\Gp}{\ctx} \mid
\GG,\clfdec{x}{T} \\
\GX &\deff \Semp \mid \GX,x
\end{align*}

The typing rules are defined in Fig.~\ref{fig:typing-terms-metaclf}. We only
show the rules related to the new constructions. We use the typing judgments
from CLF to check products over CLF types. We also need a filtering operation on
contexts, denoted $|\_|$ that keeps declarations of CLF types. It is defined by
$|\Semp|=\Semp$, $|\GG,\clfdec{x}{T}|=|\GG|,\dec{{!}x}{T}$, and
$|\GG,\Gg|=|\GG|$ for declarations $\Gg$ of Meta-CLF types and contexts.

\begin{figure}
\noindent  Kinds: \fbox{$\jkind[\GS_0]{\GS;\GG;\GX}{K}$}
$$
\infer { \jkind{\GG,\dec{\Gp}{\ctx};\GX}{K} } {
  \jkind{\GG;\GX}{\GP\dec{\Gp}{\ctx}.K} }
\qquad
\infer { \jkind{\GG;\GX,x}{K} } { \jkind{\GG;\GX}{\nab x.K} } \qquad \infer {
  \jtypeclf{\GS_0;|\GG|}{T}{\type} \\ \jkind{\GG,\dec{x}{T};\GX}{K} } {
  \jkind{\GG;\GX}{\clfpi\dec{x}{T}.K} }
$$
Types: \fbox{$\jtype[\GS_0]{\GS;\GG;\GX}{A}{K}$}
$$
\infer { \jtype{\GG,\dec{\Gp}{\ctx};\GX}{A}{\type} } {
  \jtype{\GG;\GX}{\GP\dec{\Gp}{\ctx}.A}{\type} } \qquad \infer {
  \jtype{\GG;\GX,x}{A}{K} } { \jtype{\GG;\GX}{\nab x.A}{K} }
$$
$$
\infer { \jctx{\GG;\GX}{\GD_1} \\ \jctx{\GG;\GX}{\GD_2} \\ \jsigclf{|\GG|}{\GS}
} { \jtype{\GG;\GX}{\trtype{\GD_1}{\GS}{\GD_2}}{\type} } \qquad \infer {
  \jtypeclf{|\GG|}{T}{\type} \\ \jtype{\GG,\decx{T};\GX}{A}{\type} } {
  \jtype{\GG;\GX}{\clfpi\decx{T}.A}{\type} }
$$
Contexts: \fbox{$\jctx[\GS_0]{\GS;\GG;\GX}{\GD}$}
$$
\infer { } { \jctx{\GG;\GX}{{\Semp}} }
\qquad
\infer { \jctx{\GG;\GX}{\GD} \\\\ x \in
  \GX\setminus\dom{\GD} \\ \jtypeclf{\GS_0;|\GG|,{!}\GD}{A}{\type}
   }
{ \jctx{\GG;\GX}{\GD,\dec{\cmod{x}}{A}} }
\qquad
\infer { \jctx{\GG;\GX}{\GD} \\ \dec{\Gp}{\ctx} \in \GG } {
  \jctx{\GG;\GX}{\GD,\Gp} }
$$
Terms: \fbox{$\jterm[\GS_0]{\GS;\GG;\GX}{M}{A}$}
$$
\infer { \jterm{\GG,x:A;\GX}{M}{B} } { \jterm{\GG;\GX}{\Gl x.M}{\GP x:A.B} }
\qquad \infer { \dec{c}{B}\in \GS \\ \jsp{\GG;\GX}{S}{B}{A} } {
  \jterm{\GG;\GX}{c\Sapp S}{A} } \qquad \infer { \dec{x}{B}\in \GG \\
  \jsp{\GG;\GX}{S}{B}{A} } { \jterm{\GG;\GX}{x\Sapp S}{A} }
$$
$$
\infer { \jterm{\GG,\dec{\Gp}{\ctx};\GX}{N}{A} } {
  \jterm{\GG}{\Gl\Gp.N}{\GP\dec{\Gp}{\ctx}.A} } \qquad \infer {
  \jterm{\GG,\clfdecx{T};\GX}{N}{A} } {
  \jterm{\GG;\GX}{\clflam x.N}{\clfpi\dec{x}{T}.A} } \qquad \infer
{ 
  \jterm{\GG;\GX}{\eps}{\trtype{\GD_1}{\GS}{\GD_2}} } {
  \jterm{\GG;\GX}{\{\eps\}}{\trtype{\GD_1}{\GS}{\GD_2}} }
$$
Spines: \fbox{$\jsp[\GS_0]{\GS;\GG;\GX}{S}{A'}{A}$}
$$
\infer { } { \jsp{\GG;\GX}{\Semp}{A}{A} }
\qquad \infer { \jctx{\GG;\GX}{\GD} \\
  \jsp{\GG;\GX}{S}{A[\GD/\Gp]}{B} } {
  \jsp{\GG;\GX}{\GD;S}{\GP\dec{\Gp}{\ctx}.A}{B} }
\qquad \infer { \text{$\alpha$
    fresh} \\ \jsp{\GG;\GX}{S}{A[\alpha/x]}{B} } { \jsp{\GG;\GX}{\#;S}{\nab
    x.A}{B} }
$$
$$
\infer { \jterm{\GG;\GX}{N}{A_1} \\ \jsp{\GG;\GX}{S}{A_2[N/x]}{B} } {
  \jsp{\GG;\GX}{N;S}{\GP\dec{x}{A_1}.A_2}{B} }
\qquad
\infer { \jtypeclf{\GS_0;|\GG|}{M}{T} \\ \jsp{\GG;\GX}{S}{A[\clf{M}c/x]}{B} } {
  \jsp{\GG;\GX}{\clf{M};S}{\GP\dec{x}{T}.A}{B} }
$$
Traces: \fbox{$\jterm[\GS_0]{\GS;\GG;\GX}{\eps}{\trtype{\GD}{\GS^*}{\GD'}}$}
$$
\infer { } { \jterm{\GG;\GX}{\trempty}{\trtype{\GD}{\GS^*}{\GD}} }
\qquad
\infer
{\jterm{\GG;\GX}{\eps_1}{\trtype{\GD_1}{\GS^*}{\GD}} \\
  \jterm{\GG;\GX}{\eps_2}{\trtype{\GD}{\GS^*}{\GD_2}} }
{\jterm{\GG;\GX}{\eps_1;\eps_2}{\trtype{\GD_1}{\GS^*}{\GD_2}} }
$$
$$
\infer { \dec{c}{A}\in \GS \\ \jspclf{\GS_0;{!}|\GG|,\GD_1}{S}{A}{\GD_1} } {
  \jterm{\GG;\GX}{\epat{\GD_2}{c\Sapp
      S}}{\trtype{\GD_0\cjoin\GD_1}{\GS}{\GD_0,\GD_2}} }
\qquad
\infer
{\decx{A}\in \GS \\
  \jsp{\GG;\GX}{S}{A}{\trtype{\GD_1}{\GS^*}{\GD_2}} } {\jterm{\GG;\GX}{x\Sapp
    S}{\trtype{\GD_1}{\GS^*}{\GD_2}} }
$$
\caption{Typing rules for terms and traces in Meta-CLF}
\label{fig:typing-terms-metaclf}
\end{figure}

\subsection{Safety for SSOS Specifications}
\label{sec:safety-ssos-spec}
We illustrate the use of Meta-CLF by stating and proving safety for the
language introduced in Sect.~\ref{sec:substr-oper-semant}.  We use Meta-CLF
over the signature $\GS_{\stlc}$ defining the language \stlc\ and its static
semantics.

Let us recall the type family expressing preservation of types in Meta-CLF:
\begin{multline*}
  {\sf tpres}:\clfpi\dec{t}{\tp}.~\nab d.~\nab g.~\Pi\dec{\Gp_1}{\ctx}.~\Pi\dec{\Gp_2}{\ctx}.~
\\\trtype{{!}d:\dest,\lin g:\gen~d~t}{\GS_\gen^*}{\Gp_1} \to
  \trtype{\Gp_1}{\GS_\step^1}{\Gp_2} \to
  \trtype{{!}d:\dest,\lin g:\gen~d~t}{\GS_\gen^*}{\Gp_2} \to \type
\end{multline*}

The proof proceeds by case analysis on the type
$\trtype{\Gp_1}{\GS_\step^1}{\Gp_2}$ (we follow essentially the same reasoning
as described in the proof sketch of Theorem~\ref{thm:safety-stlc}).
We have three cases, one case for each rule in $\GS_\step$.
For each case, we apply inversion on the trace of type
$\trtype{{!}d:\dest,\lin g:\gen~d~t}{\GS_\gen^*}{\Gp_1}$.

Consider the case $\stepeval$; in this case $\Gp_1$ must be of the form
$\GD,\decx{\eval~e~d_0}$ and $\Gp_2$ must be of the form
$\GD,\dec{y}{\ret~e~d_0}$, for some value expression $e$.  By inversion, in the
trace generating $\Gp_1$, there must be a step that uses $\geneval$ to
generate the declaration of $x$.  That is, the generating trace has the form
$$
X_1; \epat{\lin x}{\geneval~e~d_0~g_0~H}; X_2
$$
where $\dec{g_0}{\gen~d_0~t_0}$ for some type $t_0$ is generated by $X_1$ and
$H$ is a proof that $e$ has type $t_0$ (i.e.\ $H:\tof~e~t_0$).  Note that $X_2$ cannot consume $x$,
since $\eval$ is a terminal in the grammar defined by $\GS_\gen$.  Then, the
step generating $x$ can be permuted towards the end of the trace, so that
$X_2$ can be taken to be the empty trace $\trempty$.  To construct the trace
that generates $\Gp_2$, we only need to replace this last step by a $\genret$
step.

In Meta-CLF, we can write this case of the proof as follows, where we omit the
dependent arguments for clarity (like in LF and CLF, we expect that implicit
arguments can be reconstructed):
\begin{align*}
\sf{tpres/ret} :
{\sf tpres}~ &(X_1; \epat{\lin x}{\geneval~e~d_0~g_0~H}) \\
& (\epat{\lin y}{\stepeval~e~d_0~x~H_v}) \\
& (X_1; \epat{\lin y}{\genret~e~d_0~g_0~H~H_v})
\end{align*}
where
$X_1:\trtype{{!}d:\dest,\lin g:\gen~d~t}{\GS_\gen^*}{\Gp_1',{!}d_0:\dest,\lin g_0:\gen~d_0~t_0}$
and $\dec{H_v}{\val~e}$.
The full type making explicit all arguments is the following:
\begin{align*}
\sf{tpres/ret} ~:~& \nab x.\nab y.\nab d.\nab g.\nab d_0.\nab g_0.\Pi\Gp_1':\ctx.
                   \clfpi e:\texp. \clfpi t:\tp. \clfpi t_0:\tp. \\
            & \Pi X:\trtype{{!}d:\dest,\lin g:\gen~d~t}{\GS_\gen^*}
            {\Gp_1',{!}d_0:\dest,\lin g_0:\gen~d_0~t_0}. \\
            & \Pi \dec{H}{{\sf of}~e~t_0}.\Pi\dec{H_v}{\val~e}. \\
            & {\sf tpres}~t~d~g~(\Gp_1',{!}d_0:\dest,\lin x:\eval~e~d_0)~
                  (\Gp_1',{!}d_0:\dest,\lin y:\ret~e~d_0) \\
                  & ~~~~~~~~~(X_1; \epat{\lin x:\eval~e~d_0}{\geneval~e~d_0~g_0~H}) \\
                  & ~~~~~~~~~ (\epat{\lin y:\ret~e~d_0}{\stepeval~e~d_0~x~H_v}) \\
                  & ~~~~~~~~~ (X_1;\epat{\lin y:\ret~e~d_0}{\genret~e~d_0~g_0~H~H_v})
\end{align*}
Although we do not treat implicit argument inference in this paper, we expect
that we can infer implicit arguments by extending the LF type reconstruction
algorithm.  As we see from the case above, implicit arguments greatly increase
the usability of the system, allowing to write more concise and clear proofs.

The other two cases of the proof, corresponding to $\stepapp$ and $\stepbeta$,
are given below:
\begin{align*}
  {\sf tpres/app} : {\sf tpres}~
  & (X_1;\epat{\lin x}{\geneval~(\app~e_1~e_2)~d_0~g_0~H}) \\
  & (\epat{{!}d_1,{!}d_2,\lin x_1,\lin x_2,\lin f}{\stepapp~e_1~e_2~d_0~x}) \\
  & (X_1; \epat{{!}d_1,{!}d_2,\lin f,\lin g_1,\lin g_2}{\genfapp~d_0~g_0}; \\
  & ~~~~~~~
  \epat{\lin x_1}{\geneval~e_1~d_1~g_1~H_1};~\epat{\lin x_2}{\geneval~e_2~d_2~g_2~H_2}) \\[2ex]
  {\sf tpres/beta} : {\sf tpres}~
  &
  (X_1;\epat{{!}d_1,{!}d_2,\lin f,\lin g_1,\lin g_2}{\genfapp~d_0~g_0}; \\
  & ~~~~~~~~\epat{\lin x_1}{\genret~(\lam~e_1)~d_1~g_1~H_1~H_{v_1}};~
  \epat{\lin x_2}{\genret~e_2~d_2~g_2~H_2~H_{v_2}})\\
  & (\epat{\lin y}{\stepbeta~e_1~e_2~d_1~d_2~d~x_1~x_2~f}) \\
  & (X_1;\epat{\lin y}{\geneval~(e_1~e_2)~d_0~g_0~H};\epat{{!}d_1}{\gendest};\epat{{!}d_2}{\gendest})
\end{align*}
The latter case is the most interesting.
By inversion, the generated state $\Gp_1$ must be of the form
$$
\begin{array}{lll}
   X_1; & \epat{{!}d_1,{!}d_2,\lin g_1,\lin g_2,\lin f}{\genfapp~d_0~g_0};           & X_2;
\\      & \epat{\lin x_1}{\genret~(\lam~e_1)~d_1'~g_1'~H_1~H_{v_1}}; & X_3;
\\      & \epat{\lin x_2}{\genret~e_2~d_2'~g_2'~H_2~H_{v_2}};        & X_4
\end{array}
$$
for some traces $X_1$, $X_2$, $X_3$, $X_4$.
It must be that $d_1=d_1'$ and $d_2=d_2'$, and also $g_1=g_1'$ and $g_2=g_2'$.
Note that $X_2$, $X_3$, and $X_4$ cannot use $g_1$ and $g_2$, so the trace can be
reordered as
$$
\begin{array}{ll}
   X_1;X_2;X_3;X_4; & \epat{{!}d_1,{!}d_2,\lin g_1,\lin g_2,\lin f}{\genfapp~d_0~g_0};
\\                  & \epat{\lin x_1}{\genret~(\lam~e_1)~d_1~g_1~H_1~H_{v_1}};
\\                  & \epat{\lin x_2}{\genret~e_2~d_2~g_2~H_2~H_{v_2}})
\end{array}
$$
and $X_1$, $X_2$, $X_3$, and $X_4$ can be collapsed into one trace variable.
In the generated trace after the rewriting step, we simply replace the
generation of $\fapp$ and the two $\ret$s with a single $\eval$ fact.
We also need two $\gendest$ steps for the destinations $d_1$ and $d_2$ which are
not used anymore.

For proving progress, we first need to define a sum type that encodes the
result: either we are at a final state or we can take a step.
\begin{align*}
  \result &: \ctx \to \type \\
  \resfinal &: \trtype{\dec{!d}{\dest},\dec{\lin
      x}{\gen~d~t}}{\GS_\gen^*}{\Gp,\dec{{!}d}{\dest},\dec{\lin x}{\ret~e~d}} \to
  \result~(\Gp,\dec{\lin x}{\ret~e~d}) \\
  \resstep &: \trtype{\Gp_1}{\GS_\step^1}{\Gp_2} \to \result~\Gp_1
\end{align*}
Note in $\resfinal$ that the generated trace has only one $\ret$ fact containing
the final value at destination $d$, while $\Gp$ contains only destinations
(obtained from $\gendest$ steps) that are not used anymore.
These are destinations that were generated during the evaluation of an
expression by the $\stepapp$ rule.

The progress theorem relates a well-typed state with a result:
\begin{align*}
  {\sf progress}: \nab d.\nab g.\clfpi t:\tp.\Pi\Gp:\ctx.
  \trtype{{!}d:\dest,\lin g:\gen~d~t}{\GS_\gen^*}{\Gp}\to
  \result~\Gp\to\type
\end{align*}
We proceed by case analysis on the first step of the trace.
If it is $\genret$, then we are at a final state:
\begin{align*}
  {\sf p/ret} : {\sf progress}~(\epat{\lin x}{\genret~e~d~g~H~H_v}; X)~(\resfinal~(\epat{\lin x}{\genret~e~d~g~H~H_v}; X))
\end{align*}

If the trace starts with a $\geneval$ step, then we can make a step
depending on which expression is generated.
If the expression generated is an abstraction, we can make a step using
$\stepeval$ since abstractions are values.
If the expression generated is an application, we can make a step using
$\stepapp$.
\begin{align*}
{\textsf{p/ev-lam}} &: {\sf progress}~(\epat{\lin x}{\geneval~(\lam~e)~d~H}; X) \\
& ~~~~~~~~~~~~~~~~~(\resstep~(\epat{\lin x}{\stepeval~(\lam~e)~d~x~(\vallam~e)})) \\
{\textsf{p/ev-app}} &: {\sf progress}~(\epat{\lin x}{\geneval~(\app~e_1~e_2)~d~H}; X) \\
& ~~~~~~~~~~~~~~~~~(\resstep~(\epat{{!}d_1,{!}d_2,\lin x_1,\lin x_2,\lin f}{\stepapp~e_1~e_2~x~H}))
\end{align*}
Finally, we consider the case where the first step is $\genfapp$.
The generated trace has the form:
$$ \epat{{!}d_1,{!}d_2,\lin f,\lin g_1,\lin g_2}{\genfapp~d~g; X_1~d_1~g_1 ; X_2~d_2~g_2} $$
where the traces $X_1$ and $X_2$ generate trees rooted at $d_1$ and $d_2$
respectively.
Note that, because $\GS_\gen$ is a generative grammar, these traces are
independent.

We have three subcases: either $X_1$ and $X_2$ generate final states (in which
case we can make a step using $\stepbeta$, or we can make a step in either $X_1$
or $X_2$:
\begin{align*}
  {\textsf{p/fapp1}} &: {\sf progress}~
  (\epat{{!}d_1,{!}d_2,\lin f,\lin g_1,\lin g_2}{\genfapp~d~g};\\
  & ~~~~~~~~~~~~~~~~~~~\epat{\lin x_1}{\genret~(\lam~e_1)~d_1~g_1~H_1~H_{v_1}}; \\
  & ~~~~~~~~~~~~~~~~~~~\epat{\lin x_2}{\genret~e_2~d_2~g_2~H_2~H_{v_2}}) \\
  & ~~~~~~~~~~~~~~~~~(\resstep~(\epat{\lin y}{\stepbeta~e_1~e_2~d_1~d_2~d~x_1~x_2~f})) \\[2ex]
  {\textsf{p/fapp2}} &: {\sf progress}~
  (\epat{{!}d_1,{!}d_2,\lin f,\lin g_1,\lin g_2}{\genfapp~d~g;X_1~d_1~g_1;X_2~d_2~g_2}) \\
  &~~~~~~~~~~~~~~~~~~~(\resstep~z) \\
  & ~~\gets {\sf progress}~(X_1~d_1~g_1)~(\resstep~z) \\[2ex]
  {\textsf{p/fapp3}} &: {\sf progress}~
  (\epat{{!}d_1,{!}d_2,\lin f,\lin g_1,\lin g_2}{\genfapp~d~g;X_1~d_1~g_1;X_2~d_2~g_2}) \\
  &~~~~~~~~~~~~~~~~~~~(\resstep~z) \\
  & ~~\gets {\sf progress}~(X_2~d_2~g_1)~(\resstep~z)
\end{align*}

\paragraph{Totality.}
We showed how to encode, in Meta-CLF, proofs of safety for a small programming
language with a parallel semantics.
A natural question is: are these valid proofs? This amounts to check totality,
which is the conjunction of two properties: coverage (i.e., all cases are
considered) and termination.
While we do not give a formal treatment of totality for Meta-CLF (which we leave
for future work), we can show, informally, that both proofs given above are
total.

It is easy to see that both proofs are terminating: in the case of preservation,
the proof is not recursive, while for progress, recursive calls are performed on
smaller traces.

Checking coverage is trickier, since it encompasses checking coverage for
traces, which is a difficult problem because trace equality allows permuting
steps.

In our proof of preservation, coverage checking manifests itself in the use of
inversion: for each case of the step relation from context $\Gp_1$ to $\Gp_2$
we apply inversion to obtain a pattern that covers the generation of $\Gp_1$.
As we explained above, we only need one pattern for each case.

Note that coverage in the proof of preservation depends on the fact that
$\GS_\gen$ is a generative grammar.
In particular, that terminal symbols are not removed from the context once they are
produced, and that fresh destinations are created for each non-terminal (this property
is used in $\stepbeta$ case).

In the proof of progress, coverage checking is directly performed over the
generated trace. The interesting case is $\genfapp$, since this case involves
splitting the trace between the trees generated for each of the children of
$\fapp$.
Again, this is possible because each non-terminal is associated with a fresh
destination, so generated traces from starting from different non-terminals are
independent.

While coverage checking for traces in general is a difficult problem, by
restricting to traces generated using grammars, we expect to obtain a relatively
simple algorithm to solve this problem.

It is important to note that these proofs would be similar if we consider a
sequential semantics.  Viewed in a different direction, having a parallel
semantics does not change the proof with respect to the sequential case.  The
reason is the use of SSOS and trace equality, and the same holds true when
considering other programming constructions~\cite{Simmons12}.  The burden of
the proof is shifted to the coverage checker.  However, the use of generative
grammars makes it possible to automate coverage checking.






\section{Related Work}\label{sec:related-work}

The original reports that introduced CLF~\cite{CervesatoPWW03,CervesatoWPW03}
include many applications, including SSOS of several programming languages
features.
However, no meta-reasoning is developed.

Attempts to perform meta-reasoning within CLF have proved to be
unsatisfactory.  Watkins et al.~\cite{WatkinsCPW08} define an encoding of the
$\pi$-calculus and correspondence assertions for it in CLF\@.  They define an
abstraction relation that relates a concurrent computation with a sequence of
events.  However, due to lack of trace types, it is not possible to state the
abstraction relation. This means also that coverage is difficult to establish.
Schack-Nielsen~\cite{Schack-Nielsen11} discusses the limitations of CLF to prove
the equivalence between small-step and big-step semantics of MiniML.

Simmons~\cite{Simmons12} introduced the notion of generative grammar that
generalizes both context-free grammars and regular worlds used in LF\@.  He
describes in detail the use generative grammars and SSOS\@.  However, the
proofs of safety are done only ``on paper'' since his framework is not
expressive enough for this task.  This work is an attempt to define a logical
framework to carry out the proofs described in~\cite{Simmons12}.

Finally, let us mention our previous work on matching traces in
CLF~\cite{CervesatoPSSS12b} which provides the basis for defining a moded
operational semantics for Meta-CLF\@.  We expect to strengthen the results
given in~\cite{CervesatoPSSS12b} by restricting the matching problem to traces
produced using generative grammars.

\section{Conclusions}\label{sec:conclusions}

We have developed a logical framework for meta-reasoning about specifications
written in CLF.  We show a typical use of this framework by proving type
preservation and progress for a small programming language with a parallel
semantics.  Trace equality simplifies the proof as we do not have to worry
about proving properties about step interleavings during the parallel
execution.  However, the downside of this approach is that coverage checking
is more complicated than in the sequential case.

For future work, an immediate objective is to complete the meta-theoretical
study of Meta-CLF\@ itself. This involves proving the existence of canonical
forms, type reconstruction of implicit arguments, and totality checking
(coverage and termination). Then, of course, implementation is another obvious
objective.

This framework is well suited for the kind of proofs of safety we developed in
this paper and we expect it to perform well for other concurrent and parallel
programming constructions (e.g., futures, communication).
It will be interesting to see in what other domain we can use it.
For example, semantics of relaxed memory models, or correctness of program
transformations in the presence of threads.

\bibliographystyle{eptcs}
\bibliography{references}

\end{document}

%% file: definitions.tex
\newcommand{\eps}{\ensuremath{\varepsilon}}
\newcommand{\Gd}{\ensuremath{\delta}}
\newcommand{\Gg}{\ensuremath{\gamma}}
\newcommand{\Gl}{\ensuremath{\lambda}}
\newcommand{\Gp}{\ensuremath{\psi}}
\newcommand{\Gt}{\ensuremath{\tau}}
\newcommand{\GD}{\ensuremath{\Delta}}
\newcommand{\GG}{\ensuremath{\Gamma}}
\newcommand{\GP}{\ensuremath{\Pi}}
\newcommand{\GS}{\ensuremath{\Sigma}}
\newcommand{\GX}{\ensuremath{\Xi}}

\newcommand{\MA}{\ensuremath{\mathcal{A}}}

\newcommand{\nab}{\ensuremath{\nabla}}

\newcommand{\lolli}{\multimap}
\newcommand{\eeq}{\equiv}

\newcommand{\stlc}{\ensuremath{{\lambda^{\rightarrow}}}}

\newcommand{\rewto}{\leadsto}

\newcommand{\deff}{\mathrel{::=}}
\newcommand{\trtype}[3]{\{{#1}\}\,{#2}\,\{{#3}\}}

\newcommand{\type}{\mathsf{type}}
\newcommand{\kind}{\mathsf{kind}}
\newcommand{\ctx}{\mathsf{ctx}}
\newcommand{\jud}[4][]{{#2}\vdash_{#1}{#3}:{#4}}
\newcommand{\jchk}[4][]{{#2}\vdash_{#1}{#3}\Leftarrow{#4}}

\newcommand{\sjud}[3][]{{#2}\vdash_{#1}{#3}}
\newcommand{\jkind}[3][]{{#2}\vdash_{#1}{#3}:\kind}
\newcommand{\jtype}[4][]{{#2}\vdash_{#1}{#3}:{#4}}
\newcommand{\jtypeclf}[3]{{#1}\vdash_{\mathrm{CLF}}{#2}:{#3}}
\newcommand{\jterm}[4][]{{#2}\vdash_{#1}{#3}:{#4}}
\newcommand{\jsp}[5][]{{#2}\vdash_{#1}{#3}:{#4}>{#5}}
\newcommand{\jspclf}[4]{{#1}\vdash_{\mathrm{CLF}}{#2}:{#3}>{#4}}
\newcommand{\jctx}[3][]{{#2}\vdash_{#1}{#3}~\ctx}

\newcommand{\jsigclf}[2]{{#1}\vdash_{\mathrm{CLF}}{#2}}

\newcommand{\clfpi}{\widehat{\Pi}}
\newcommand{\clflam}{\widehat{\lambda}}
\newcommand{\clfdec}[2]{{#1}{\hat{:}}{#2}}
\newcommand{\clfdecx}[1]{\clfdec{x}{#1}}

\newcommand\epat[2]{\{{#1}\}{\shortleftarrow}{#2}}
\newcommand{\trempty}{\diamond}

\newcommand\cjoin{\bowtie}
\newcommand{\dec}[2]{{#1}{:}{#2}}
\newcommand{\decx}[1]{\dec{x}{#1}}
\newcommand{\dom}[1]{\textrm{dom}({#1})}

\newcommand{\FV}[1]{\textrm{FV}({#1})}

\newcommand{\clf}[1]{\langle{#1}\rangle}

\newcommand{\val}{\mathsf{value}}
\newcommand{\vallam}{\mathsf{value/lam}}

\newcommand{\Sapp}{\cdot}
\newcommand{\Semp}{{\cdot}}
\newcommand{\lin}{\ensuremath{{\downarrow}}}

\newcommand{\texp}{\mathsf{exp}}
\newcommand{\arr}{\mathsf{arr}}
\newcommand{\tof}{\mathsf{of}}
\newcommand{\toflam}{\mathsf{of/lam}}
\newcommand{\tofapp}{\mathsf{of/app}}
\newcommand{\lam}{\mathsf{lam}}
\newcommand{\app}{\mathsf{app}}
\newcommand{\eval}{\mathsf{eval}}
\newcommand{\ret}{\mathsf{ret}}
\newcommand{\fapp}{\mathsf{fapp}}
\newcommand{\gen}{\mathsf{gen}}
\newcommand{\step}{\mathsf{step}}
\newcommand{\dest}{\mathsf{dest}}
\newcommand{\tp}{\mathsf{tp}}
\newcommand{\stepeval}{\mathsf{step/eval}}
\newcommand{\stepapp}{\mathsf{step/app}}
\newcommand{\stepbeta}{\mathsf{step/beta}}

\newcommand{\geneval}{\mathsf{gen/eval}}
\newcommand{\genret}{\mathsf{gen/ret}}
\newcommand{\genfapp}{\mathsf{gen/fapp}}
\newcommand{\gendest}{\mathsf{gen/dest}}
\newcommand{\result}{\mathsf{result}}
\newcommand{\resfinal}{\mathsf{res/final}}
\newcommand{\resstep}{\mathsf{res/step}}

\newcommand\ein[1]{{\bullet}{#1}}
\newcommand\eout[1]{{#1}{\bullet}}

\newcommand{\narrowbox}[2]
  {\ensuremath{\fboxsep=0pt\fboxrule=0.07ex\fbox{\rule{0em}{#1}\hspace{#2}}}}
\newcommand{\nbox}{\mathopen{\narrowbox{1.75ex}{0.35em}}\hspace*{0.1em}}
\newcommand{\Mod}{\mathopen{\nbox}} 

\newcommand\cmod[1]{{\Mod}{#1}}
\newcommand\cmodname{{\Mod}}
